\documentclass[10 pt]{amsart}
\usepackage{amsmath}
\usepackage{mathtools}
\usepackage{amssymb}
\usepackage{amsthm}
\usepackage{fancyhdr}
\usepackage[all]{xy}
\usepackage{hyperref}
\usepackage{paralist}
\usepackage{tikz}
\usepackage{algorithm}
\usepackage{algorithmic}
\usepackage{graphicx}
\usepackage{url}

\usepackage{multirow}

\newtheorem{theorem}{Theorem}[section]
\newtheorem{lemma}[theorem]{Lemma}
\theoremstyle{proposition}
\newtheorem{proposition}[theorem]{Proposition}
\theoremstyle{corollary}

\theoremstyle{definition}
\newtheorem{definition}[theorem]{Definition}

\theoremstyle{remark}
\newtheorem{remark}[theorem]{Remark}

\numberwithin{equation}{section}

\newcommand{\F}{\mathbb F}

\begin{document}

\email{caglar21@itu.edu.tr, nari15@itu.edu.tr, ozdemiren@itu.edu.tr}
\title{An Application of Nodal Curves  }
\author{Selin Caglar}
\author{Kubra Nari}
\author{Enver Ozdemir}
\address{Informatics Institute, Istanbul Technical University}

\begin{abstract}

 In this work, we present an efficient method for computing in the generalized Jacobian of special singular curves, nodal curves. The efficiency of the operation is due to the representation of an element in the Jacobian group by a single polynomial. In addition, we propose a probabilistic public key algorithm  as an application of nodal curves.

\end{abstract}
\thanks{Classification: 11G20, 11Y99, 94A60 }
\maketitle

\specialsection*{\bf \Large Introduction}
\indent The Jacobian groups of smooth curves, especially for those belonging to the elliptic and hyperelliptic curves, have been rigorously investigated \cite{Anderson,Cantor, Mumford} due to their use in computational number theory and cryptography \cite{CohFrey, KOBEL, KOBHYP,HLEN, VMIL}. Even though the singular counterparts of these curves have simple geometric structures, the generalized Jacobian groups of these curves might be potential candidates for further applications in computational number theory and cryptography.

\indent An element in the Jacobian of a hyperelliptic curve is represented by a pair of polynomials $(u(x),v(x))$ satisfying certain conditions \cite{Mumford}. The situation is the same for higher degree curves. For example, an element $D$ in the Jacobian of a superelliptic curve $S:y^3=g(x)$ is represented by  a triple  of  polynomials $(s_1(x),s_2(x),s_3(x))$ satisfying certain conditions \cite{Bauer}. Therefore, for a smooth curve, we do not have the liberty to choose any polynomial $u(x)$ and say that it is a coordinate of an element in the Jacobian group of a given curve. On the other hand, we show that one can treat almost any polynomial $h(x)$ as an element of the Jacobian of a nodal curve. In other words, a random element in a generalized Jocabian of a singular curve can easily be selected which eventually might encourage researchers to work with these curves for further applications in the related areas in addition to \cite{OzdemirPOL}. In this respect, we present an application where nodal curves are employed towards the construction of a probabilistic public-key cryptosystem.\\
\indent In the first part of the paper, we present an efficient method to perform group operation in the Jacobians of nodal curves based on the work \cite{OzdemirPHD}. The method is basically a modification of Mumford representation\cite{Mumford} and Cantor's algorithm\cite{Cantor}. We note that in the work \cite{OzdemirPHD}, Mumford representation and Cantor's algorithm are extended for general singular curves. For our purposes, a nodal curve $N$ over a finite field $\F_q$ with a characteristic $p\ne 2$ is a curve defined by an equation $y^2=xf(x)^2$ where $f(x)\in \F_q[x]$ is an irreducible polynomial. Let $d=\deg(f(x))$. We show that almost any polynomial $h(x)$ with $\deg(h(x))<d$ uniquely represents an element $D$ in the Jacobian of the curve. Then, we define an addition algorithm for this single polynomial representation in the Jacobian group. The representation provides advantageous in practical applications as the implementation results are illustrated at the end of each section. \\
\indent The digital communication security is ensured via cryptographic primitives. The vulnerabilities of these primitives are based on some mathematical problems. One of the most popular and practical public-key cryptosystem Rivest-Shamir-Adleman (RSA) exploits a group structure in the multiplicative group $(\mathbb Z^*_n,\cdot)$ where $n$ is a multiple of two prime integers $p$ and $q$. The public key of a user is just $(n,e)$ where $e$ is a random integer coprime to $\phi(n)=(p-1)(q-1)$.  The security of this algorithm relies on the hardness assumption of finding factors of $n$ or finding an $e^{th}$ root of a random element in $(\mathbb Z_n^*,\cdot)$. Breaking the RSA without factoring the RSA modulus $n$ is called the RSA problem. In the second part of the paper, we present a public-key algorithm whose security is again based on the hardness assumption of integer factorization. The proposed algorithm might be a candidate in case the RSA problem has a solution. The proposed public-key algorithm with nodal curves might emphasize  the use of such curves in practical applications.

\section{Singular  Curves}
\indent  The Jacobian is an abstract term which attaches an abelian group to an algebraic curve. This abstract group, Jacobian, is simply the ideal class group of the corresponding coordinate ring. If the curve is smooth, the attached group is called Jacobian, otherwise it is called  Generalized Jacobian \cite{Rosenlicht}. However, we will keep using the term `Jacobian' for all kinds of curves. We are only interested in computing in the Jacobian groups of nodal curves and more details about algebraic and geometric properties of these curves can be found in \cite{Bosch,LIU}. As we mentioned above, for our purposes, a nodal curve is defined by an equation of the form $N:y^2=xf(x)^2$ over a field $\mathbb F_q$ with a characteristic different from 2 where $f(x)$ is an irreducible polynomial in $\mathbb F_q[x]$. The attached  Jacobian group is denoted by Jac($N$). For example, if the degree of $f(x)$ is 1, that is $N:y^2=x(x+a)^2$ for some $a\ne 0\in \F_q$, computing in the Jacobian group is similar to computing in an elliptic curve group \cite[Section 2.10]{WAS}. In order to perform group operation for a curve, each element in the Jacobian should be represented in a concrete way. The Mumford representation provides a concrete representation for elements in the Jacobians of hyperelliptic curves. This representation has been extended\cite{OzdemirPHD} for singular curves defined by equations of the form $y^2=g(x)$. Below, we present Mumford representation along with Cantor's algorithm which provides a method of computing in the Jacobians for aforementioned singular curves \cite{OzdemirPHD}.
\subsubsection{The Mumford Representation}   \label{mumford}
\indent Let $f(x) \in \mathbb F_q$ be a monic  polynomial of degree $2g+1$  such that $g\ge 1$. A curve $H$ over $\F_q$ is defined by  the equation $y^2=f(x)$.  Any divisor class $D$ in the Jacobian group of $H$, Jac($H$), is represented by a pair of polynomials $[u(x),v(x)]$ satisfying the following:

\begin{enumerate}
\item $\deg(v(x))<\deg(u(x))$.
\item $v(x)^2-f(x)$ is divisible by $u(x)$. 
\item If $u(x)$ and $v(x)$ are both multiples of $(x-a)$ for a singular point $(a,0)$ then
$\dfrac{f(x)-v(x)^2}{u(x)}$ is not a multiple of $(x-a)$. Note that $(a,0)$ is a singular point of $H$ if $a$ is a multiple root of $f(x)$.\\
\end{enumerate}

Any divisor class $D\in$ Jac($H$) is uniquely represented by a (reduced) pair $(u(x),v(x))$ if in addition to the above properties, we have:
\begin{enumerate}
\item $u(x)$ is monic.
\item deg$(v(x))<$deg$(u(x))\leq g.$
\end{enumerate}
 We should note here that the identity element is represented by $[1,0]$.
\subsubsection{Cantor's Algorithm} 
This algorithm takes two divisor classes $D_1=[u_1(x),v_1(x)]$ and $D_2=[u_2(x),v_2(x)]$ on $H: y^2=f(x)$ and outputs the unique representative for the divisor class $D$ such that $D=D_1+D_2$. 
\begin {enumerate}
\item $h= \gcd (u_1,u_2,v_1+v_2)$ with polynomials $h_1, h_2,h_3$ such that\\ $h=h_1u_1+h_2u_2+h_3(v_1+v_2)$\\
\item $u=\dfrac {u_1u_2}{h^2}$ and $v\equiv \dfrac{h_1u_1v_2+h_2u_2v_1+h_3(v_1v_2+f)}{h}$ (mod $u$)\\
\\
{\bf repeat:}
\\
\item $\widetilde{u}=\dfrac{v^2-f}{u}$ and $\widetilde{v}\equiv v$ (mod $\widetilde u$)\\

\item $u=\widetilde u$ and $v=-\widetilde v$
\\
{\bf until} deg $(u)\leq g$ 

\item Multiply $u$ by a constant to make $u$ monic.
\item $D=[u(x),v(x)]$
\end{enumerate}
The combination of the third and fourth steps is called the reduction steps which eventually return a unique reduced divisor for each class. The justification of the above statements is given in the \cite{OzdemirPHD}.\\

\subsection{Nodal Curves}
A nodal curve  $N$ over a field is an algebraic curve with finitely many singular points which are all simple double points. The curve $N$ has a smooth resolution $\widetilde{N}$ obtained by separating the two branches at each node. In this section, we are  going to construct a  representation for elements in the Jacobians of  nodal curves. Again, note that the curves under consideration are of the form $N:y^2=xf(x)^2$ where $f(x)$ is an irreducible polynomial of degree $d$ over the field $\mathbb F_q$. Here, we briefly mention related results for nodal curves especially from the work of M. Rosenlicht \cite{Rosenlicht, Rosenlicht54}. Let 
$C$ be a smooth algebraic curve and $\mathfrak m$ be a modulus, i.e. $\mathfrak m=\sum_{P\in C} m_PP$ where $m_P$ is non-negative. Let denote the generalized Jacobian group of $C$ with respect to the modulus $\mathfrak m$ by $J_{\mathfrak m}(C)$. We have a surjective homomorphism\cite{Rosenlicht54} 
$$\sigma : J_{\mathfrak m}(C)\rightarrow Jac(C).$$  
\begin{remark}\label{Rm1} The normalization of the nodal curve $N$ gives $\mathbb P^1$ so we take $C=\mathbb P^1$. It is known that Jac($C)$ is trivial. In our case i.e., $J_{\mathfrak m}(C)=$Jac($N)$, the kernel of $\sigma$ is isomorphic to a torus $\mathbb G_m^d$ of dimension $d$=deg($f(x)$). Note that, the modulus $\mathfrak m$ has only singular points which are the roots of $f(x)$. See  \cite{Dechene,OzdemirPOL,Serre} for more details. 
\end{remark}	
%Let $N:y^2=xf(x)^2$ be a nodal curve over $\mathbb F_q$. Let $\deg(f(x))=d$. By the work of \cite{OzdemirPHD}, we know that there is a unique reduced pair $(u(x),v(x))$ with $\deg(u(x))\le d$ for each class in Jac($N$). This is due to extension of Mumford Representation for singular curves. Consider a pair $(f^2(x),h(x)f(x))$ where $h(x)$ is a polynomial such that $\deg(h(x))\le \deg(f(x))$ and $\gcd(x-h^2(x),f(x))=1$.
\begin{theorem} \label{theorem1}
Let $f(x)$ be an irreducible polynomial of degree $d$ over $\F_q$ and $N:y^2=xf^2(x)$ be a nodal curve over $\mathbb F_q$. Any divisor class $D\in$ Jac($N$) is uniquely represented by a  polynomial $h(x)$ satisfying $$\deg(h(x))< d \text{ and }  \gcd(f(x),x-h^2(x))=1.$$  
\end{theorem}

We are going to prove this theorem by a series of lemma.

\begin{lemma} \label{Lm1}
Let $N$ be as above. Let $h(x)$ be a polynomial of degree less than $d$ such that $\gcd(f(x),x-h^2(x))=1$. Then, the pair $D=[f^2(x),h(x)f(x)]$ represents an element in Jac($D$).
\end{lemma}

\begin{proof}
Let $$D=[u(x),v(x)]=[f^2(x),h(x)f(x)].$$ Both $u(x)$ and $v(x)$ are divisible by $x-a$ where $a$ is any root of $f(x)$ over the algebraic closure of $\mathbb F_q$. On the other hand, $\gcd(f(x),x-h^2(x))=1$ so $$\dfrac{xf^2(x)-v^2(x)}{u(x)}=\dfrac{xf^2(x)-h^2(x)f^2(x)}{f^2(x)}=x-h^2(x)$$ is not divisible by $x-a$ for any root $a$ of $f(x)$. By Mumford Representation which is defined above, $[f^2(x),h(x)f(x)]$ represents an element $D$ in Jac($N$).
\end{proof}

\begin{lemma}\label{Lm2}
Let $\gcd(f(x),x-h_i^2(x))=1$ and $\deg(h_i(x))<d$ for each $i=1,2$. Let $D_1=[f^2(x),h_1(x)f(x)]$ and $D_2=[f(x)^2,h_2(x)f(x)]$ be two divisor classes.
We find $$D_1+D_2=D_3=[f^2(x),h_3(x)f(x)]$$ via
\begin{enumerate}
\item finding two polynomials $g_1(x),g_2(x)$ such that $$g_1(x)f(x)+g_2(x)(h_1(x)+h_2(x))=1\\$$
\item Then computing $$h_3(x)\equiv (f(x)h_1(x)g_1(x)+g_2(x)(h_1(x)h_2(x)+x)) \mod f(x)$$ with $\deg(h_3(x))<d$.\\
\end{enumerate} 
\end{lemma}

\begin{proof}
We apply Cantor's Algorithm for $D_1+D_2$ to confirm the addition algorithm. 
\begin{enumerate}
\item We first compute:\\

$\begin{array}{lll}
       \gcd(f(x)^2,f(x)^2,h_1(x)f(x)+h_2(x)f(x))&=& f(x)\cdot \gcd(f(x),f(x),h_1(x)+h_2(x))\\
       \\
       &=& f(x)\cdot \gcd(f(x),h_1(x)+h_2(x))\\
       \\
       &=&f(x)
      \end{array}$
 
with $g_1(x),g_2(x)$ such that $g_1(x)f(x)+g_2(x)\Big(h_1(x)+h_2(x)\Big)=1.$  \\

\item Set\\
$\begin{array}{lllll}
u_3(x)&=&\dfrac{f(x)^2f(x)^2}{f(x)^2}=f(x)^2 \\
\\
v_0(x)&=&\dfrac{g_1(x)f^3(x)h_1(x)+g_2(x)\Big(h_1(x)h_2(x)f^2(x)+xf(x)^2\Big)}{f(x)}\\
\\
&&=g_1(x)f^2(x)h_1(x)+g_2(x)\Big(h_1(x)h_2(x)f(x)+xf(x)\Big).\\
\\
\end{array}$
\item Then\\
$\begin{array}{lllll}
v_3(x)&\equiv& v_0(x) \mod u_1(x)\\
\\
&\equiv& g_1(x)f(x)^2h_1(x)+g_2(x)\Big(h_1(x)h_2(x)f(x)+xf(x)\Big)\mod u_3(x)=f(x)^2\\
\\
&=&\underbrace{\Big(f(x)g_1(x)h_1(x)+g_2(x)(h_1(x)h_2(x)+x) \mod f(x)\Big)}_{h_3(x)}f(x)\\
\\
&=& h_3(x)f(x) \text{ with} \deg(h_3(x))<d\\
\\
\end{array}
$
\item $D_1+D_2=[u_3(x),v_3(x)]=[f^2(x),h_3(x)f(x)]=D_3$\\
\end{enumerate}

\end{proof}

Note that $$[f^2(x), h(x)f(x)]+[f^2(x),-h(x)f(x)]=[1,0].$$

\begin{lemma}
Let $N:y^2=xf^2(x)$ be a nodal curve over $\mathbb F_q$ such that $f(x)$ is an irreducible polynomial. Let $$\begin{array}{ccc}
D_1 & = & [f^2(x),\quad h_1(x)f(x)] \text{ with } \deg (h_1(x))<\deg(f(x))\\
D_2 &= & [f^2(x),\quad h_2(x)f(x)] \text{ with } \deg(h_2(x))<\deg(f(x))
\end{array}$$ 
such that $$h_1(x)\ne h_2(x).$$
Then $$D_1\neq D_2.$$ 
\end{lemma}

\begin{proof}
Suppose $$D_1=D_2$$ then 
$$\begin{array}{ccl}
[1,0]&=& D_1+(-D_2)\\
&=& [f^2(x),h_1(x)f(x)]+[f^2(x),-h_2(x)f(x)]\\
\end{array}$$
This is possible only when $h_1(x)+(-h_2(x))$ is zero or a multiple of $f(x)$. Note that it can not be a multiple of $f(x)$ as the degrees of both $h_1(x)$ and $h_2(x)$ are less than $\deg(f(x))$. Therefore, as long as  $h_1(x)\ne h_2(x)$, we do not get $D_1=D_2$.
\end{proof}

\noindent{\it Proof of Theorem \ref{theorem1}}:\\
In Lemma \ref{Lm1}, we defined a new type of a representation for  elements in the Jacobian group of $N:y^2=xf^2(x)$, i.e., each element is represented by a pair $[f^2(x),h(x)f(x)]$ such that $\deg(h(x))<\deg(f(x))$ and $f(x)$ doesn't divide $x-h^2(x)$. The lemma \ref{Lm2} shows how to perform the group operation with this representation. In the last lemma, we showed that for distinct $h(x)$, the pairs represent distinct elements in the Jacobian group. As the degree of $h(x)$ is less than $d$, we have approximately $q^{\deg(f(x))}$ such pairs which is equal to the order of the Jacobian group by the remark \ref{Rm1} and this completes the proof.  \\
\indent  Let $\mathbb F_q$ be a finite field with a characteristic $p\ne2$. Let $N:y^2=xf^2(x)$ be a singular curve such that $f(x)$ is an irreducible polynomial of degree $d$ over $\mathbb F_q$. The above discussion leads us to the following Algorithm \ref{alg1}.
 
\begin{algorithm}[h]                  
	\caption{Addition algorithm in the Jacobian group of the curve $N:y^2=xf^2(x)$ over $\mathbb F_q$. }           
	\label{alg1}                     
	\begin{algorithmic}[1]   
		\renewcommand{\algorithmicrequire}{\textbf{Input:}}
		\renewcommand{\algorithmicensure}{\textbf{Output:}}
		\REQUIRE $D_1$ and $D_2$ represented by $h_1(x), h_2(x)$ respectively such that $\deg h_1(x),h_2(x) <\deg f(x)=d$
		\ENSURE $D= D_1 + D_2=h(x)$
		\STATE If $h_1(x)+h_2(x)\equiv 0 \mod f(x)$ set $h(x)=[1,0]$ (identity). Otherwise do:
		\STATE Find $g_1(x)$ and $g_2(x)$ such that
		$g_1(x)f(x)+g_2(x)(h_1(x)+h_2(x))=1$.
		\STATE Set: $$h(x)\equiv (g_2(x)(h_1(x)h_2(x)+x)) \mod f(x)$$
		
		\RETURN $h(x)$\\
	\end{algorithmic}
\end{algorithm}

\begin{remark}
	We should note here that the work \cite{kohel} discusses computing in the generalized Jacobian group of nodal curves. Imitating singular cubics, the work assumes all pairs satisfying Mumford representation for smooth curve, also represents a point in the Jacobian of nodal curves. However as described in \cite{ozdemir2021factoring} extension of Mumford representation of singular curves requires additional conditions for the pairs to  represent a point in the Jacobian group of the nodal curves as described in the subsection \ref{mumford} above.
\end{remark}

 We form the curve $N$ with an irreducible polynomial $f(x)$ of degree $d$ over $\mathbb F_q$. Any polynomial $h(x)$ of degree less than $d$ with $\gcd(f(x),x-h^2(x))=1$ represents a unique element in 
Jac($N$). For two elements $D_1,D_2\in $ Jac($N$) represented by polynomials $h_1(x)$ and $h_2(x)$ respectively, we define an addition operation involving only univariate polynomial arithmetics. The algorithm returns a polynomial $h(x)\in\mathbb F_q[x]$ which uniquely represents $D=D_1+D_2$. We also note that almost all polynomials $h(x)$ of degree less than $d$ represents an element in Jac($N$) and this allows one to easily select  a random element $D$ in Jac($N$). The single polynomial representation does not only give us the liberty to select any polynomial, but it also provides an efficient group operation in the Jacobian group. The following table compares this group operation with Cantor's algorithm. According to the results in this table, the single representation of Jacobian elements has advantages over polynomial pairs representation. The time is measured while computing a $pQ$ where $Q$ is an element in the Jacobian group of $N$ and 512 bits and 1024 bits prime numbers are used for $p$. The curve is over the field $\mathbb F_p$ and the arithmetic genus of curves are the same for each comparison. The prime number p was kept fixed for each bit size.
\\

%\begin{table}[H]
%	\centering
%	\caption{Comparison of presented group operation and cantor algorithm in terms of executing time.}

%\begin{tabular}{ |p{4cm}|p{3cm}|p{3.5cm}|  }

%	\hline
%	{\bf Degree of the $g(x)$ when $N:y^2=g(x)$}& {\bf Nodal Curves (Second)} &{\bf Cantor's Algorithm (Second)}\\
%	\hline
%	5   & 0.003    &0.023\\
%	\hline
%	11&   0.01  & 0.081  \\
%	\hline
%	23 &0.019 & 0.357\\
%	\hline
%	47    &0.06 & 2.15\\
%	\hline
%	53&   0.068 & 2.98\\
%	\hline
%	63& 0.089 & 4.96   \\
%	\hline
%	71& 0.12  & 6.95\\
%	\hline
%	83   & 0.15    &10.98\\
%	\hline
%	95&   0.19  & 16.67  \\
%	\hline
%	110 &0.24 & 26.36\\
%	\hline
%	130   &0.33 & 45.46\\
%	\hline
%	145&   0.41  & 64.9\\
%	\hline
%	150& 0.43  & 72.3   \\
%	\hline
%	165& 0.52  & 100.39\\
%	\hline
%	193 & 0.69 &167.29\\
%	\hline
	
%\end{tabular}
%\end{table}

\begin{table}[H]
	
	\caption{Comparison of presented group operation and cantor algorithm in terms of executing time. }
	
	\begin{tabular}{ |p{1.5cm}|p{3.4cm}|p{2.7cm}|p{3.5cm}|  }
		
	%	\begin{tabular}{ |c|c|c|c|  }

		%\multicolumn{3}{|c|}{} \\
		\hline
		{\bf Size of $p$ (Bits)}&{\bf The degree of $g(x)$ where a curve is defined by $y^2=g(x)$}& {\bf Nodal Curves (Second)} &{\bf Cantor's Algorithm (Second)}\\
		\hline
		\multirow{15}{1.5cm}{\centering  512}  & \centering  11 & \centering  0.065    &   	0.685 \\  \cline{2-4}

		& \centering 23 &\centering   0.14 &   2.54 \\  \cline{2-4}
		
		& \centering 47& \centering 0.37  &   12.2  \\ \cline{2-4}
		
		& \centering 95 &\centering   1.12 &  37.1 \\ \cline{2-4}
		
		& \centering 107& \centering  1.4  &  45.85  \\ \cline{2-4}
		
		& \centering 127 &\centering   1.83 & 60.72 \\ \cline{2-4}
		
		& \centering 143& \centering  2.32  & 80.79 \\ \cline{2-4}
		
		& \centering 167 &\centering   2.97 & 112.83  \\ \cline{2-4}
		
		& \centering 191& \centering  4.09  & 148.72 \\ \cline{2-4}
		
		& \centering 221&\centering  5.19 & 188.72  \\ \cline{2-4}
		
		& \centering 261& \centering  7.22  & 257.06 \\ \cline{2-4}
		
		& \centering 291 &\centering   8.88 & 329.76 \\ \cline{2-4}
		
		& \centering 301& \centering  9.32  & 356.96 \\ \cline{2-4}
		
		& \centering 331 &\centering   11.1 & 425.18  \\ \cline{2-4}
		
		&\centering  387 &\centering   14.4 & 616.87 \\ 
		\hline
		\multirow{15}{1.5cm}{ \centering  1024}  & \centering  11 & \centering 0.23     & 2.44 \\ \cline{2-4}
		
			& \centering 23 &\centering   0.59 & 11.01 \\ \cline{2-4}
		
		& \centering 47& \centering  1.78  & 59.75 \\ \cline{2-4}
		
		& \centering 95 &\centering   6.07 & 111.7  \\ \cline{2-4}
		
		& \centering 107& \centering 7.59  & 195.43 \\ \cline{2-4}
		
		& \centering 127 &\centering   9.98 & 293.71 \\ \cline{2-4}
		
		& \centering 143& \centering 12.62  & 382.44 \\ \cline{2-4}
		
		& \centering 167 &\centering   16.75 &  529.04 \\ \cline{2-4}
		
		& \centering 191& \centering  23.01  & 689.51 \\ \cline{2-4}
		
		& \centering 221&\centering   29.68 & 884.34 \\ \cline{2-4}
		
		& \centering 261& \centering  42.47  & 1206.79 \\ \cline{2-4}
		
		& \centering 291 &\centering   51.28 & 1546.82 \\ \cline{2-4}
		
		& \centering 301& \centering  54.53  & 1654.85 \\ \cline{2-4}
		
		& \centering 331 &\centering   63.75 &  2008.26 \\ \cline{2-4}
		
		&\centering  387 &\centering   85.69 & 2738.57 \\ 
		
		\hline
		
	\end{tabular}

\end{table}

The tests were run on a Windows 10 OS computer with 16 GB RAM and a Intel Core i7- 10875H 2.3 GHz processor. We use the programming environment of Python with a SageMath library \cite{Sage}. 

\section{A Public-Key Algorithm}

Public-key infrastructure (PKI) is a composition of services and protocols that provides key generation and management for  public-key algorithms which are part of asymmetric encryption methods that are employed for reliable communication. In a public-key cryptography, each user has 2  keys; one is for encryption and the other one is for decryption of  messages. The encryption key is public and broadcast but the other one, the decryption key, must be kept private.
As shown in Figure \ref{fig:fig1}, the exchange of data between two users basically occurs in the following way:

\begin{itemize}
	\item The sender and recipient have their key pairs. The sender uses the recipient's public key for encryption and sends the ciphertext to the recipient.
	\item The recipient gets the ciphertext and decrypts the messages using its own private key.      
\end{itemize}

\begin{figure}[H]
	\includegraphics[width=\linewidth]{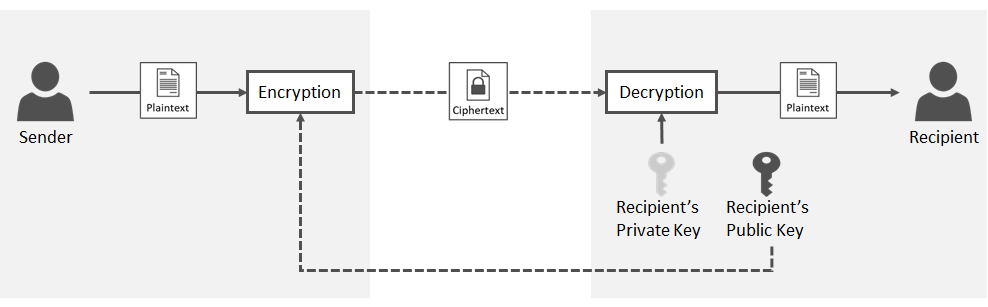}
	\caption{Public-Key Algorithm: Different keys are used for encryption and decryption.}
	\label{fig:fig1}
\end{figure}

One of the most widely used public-key encryption methods is Rivest-Shamir-Adleman (RSA) algorithm. The method is being employed in many areas like authentication and digital signature to ensure security in the information exchange systems \cite{RSA}.  Each party has a public key pair $(n,e)$ where $n$ is formed by the multiplication of prime numbers of the same size. Each of these prime integers is recommended to be at least 1024-bit number \cite{NIST}. The workflow starts with selecting a pair of primes $p,q$ satisfying certain conditions \cite{Boneh}. Once deciding such numbers then the first component of the public key is assigned to be $$n=pq$$
In the second step, an encryption key $e$ which is co-prime to $(p-1)(q-1)$ is chosen. As a result, the public key is just the pair $(n,e)$. The private decryption key is an integer $d$ such that $$de\equiv 1 \mod (p-1)(q-1)$$ Key generation steps of RSA algorithm is given in Algorithm \ref{alg:rsa_key} .

\begin{algorithm}[!h]
	\caption{RSA Key Generation}\label{alg:rsa_key}
	\begin{algorithmic}[1]
		\renewcommand{\algorithmicrequire}{\textbf{Input:}}
		\renewcommand{\algorithmicensure}{\textbf{Output:}}
		\REQUIRE $p$ and $q$ are the same size prime integers.
		\ENSURE Public Key $\gets (n,e)$, Private Key $\gets (n,d)$
		\STATE $n \gets pq$
		\STATE $\phi(n) \gets (p-1)(q-1)$
		\STATE Select $e; \qquad   \gcd(e,\phi(n),e)=1;    1< e < \phi(n) $
		\STATE Compute $d; \qquad ed=1 \mod \phi(n)$ \\
		Public Key $\gets (n,e)$\\
		Private Key $\gets (n,d)$

	\end{algorithmic}
	
\end{algorithm}

Note that the multiplicative group $G=(\mathbb Z^*_n,\cdot)$ where $\mathbb Z^*_n$ represents the numbers in between 0 and $n-1$ which are co-prime to $n$. The group $G$ has order $$\phi(n)=(p-1)(q-1)$$ and this fact allows to say $$a^{k\phi(n)+1}\equiv a \mod n$$ for any integer $k$ and any element $a$ of $G$. In RSA algorithm, any message is represented by an element $m$ of $G$. The message $m$ is encrypted by performing the following operation 
$$c\equiv m^e \mod n $$ where $c$ stands for the corresponding ciphertext. Encryption and decryption algorithms for RSA is given in Algorithm \ref{alg:rsa} .

\begin{algorithm}[!h]
	\caption{RSA Encryption and Decryption}\label{alg:rsa}
	\begin{algorithmic}[1]
		\renewcommand{\algorithmicrequire}{\textbf{Encryption:}}
		\renewcommand{\algorithmicensure}{\textbf{Decryption:}}
		\REQUIRE $E(m): $ \\
		$\quad\quad\quad\quad\qquad  c=m^e \mod n $ \\
		$\quad\quad\quad\quad\qquad $The cipher $\gets c$
		
		\ENSURE  $D(c): $ \\
		$\quad\quad\quad\quad\qquad m=c^d \mod n $ \\
		$\quad\quad\quad\quad\qquad $The message  $\gets m$

	\end{algorithmic}
	
\end{algorithm}

 If one knows the factors $p,q$ of $n$ then one can compute $\phi(n)$ and $d$ such that $de\equiv 1\mod \phi(n)$. In other words, the cipher $c$ can be converted to $m$ by anyone having the factors $p$ and $q$, that is, factoring the RSA moduli is sufficient to break the algorithm. On the other hand, one might return back to the plaintext from publicly known information $n, e, c$ without factoring $n$.  
 
 \begin{definition}
 The problem of  reaching $m$ without finding the factors of the RSA modulus is called {\it the RSA problem.} 
\end{definition}

Even though most of the research toward the security of RSA focuses on the factorization problem, the recent work has shown that converting $c$ to $m$ without factorization is more promising method \cite{Naccache}. The following algorithm might stand as a nice choice in case a practical method for the RSA problem is presented.

%In this case, the main concern is factoring such large numbers. Hence,  the security of the RSA cryptosystem is due to the fact that these large numbers cannot be factored. 

\section{The Method}

Let $N: y^2=xf^2(x)$ be a nodal curve defined over a ring $(\mathbb Z_n,+,\cdot)$. Note that, the computing in the Jacobian is discussed above while defining curves over  finite fields. Although we now define the curve $N$ over a ring, we are just going to use group operation in the Jacobian and as in the case of hyperelliptic curves, computing in the Jacobian over the ring $\mathbb Z_n$ is handled in the same manner \cite{lenstra1993}. In other words, for our purposes defining the curve modulo prime or composite will almost be the same. The integer $n$ as in the case of RSA algorithm is a composite number that is formed by multiplying  two prime numbers $p$ and $q$. The polynomial $f(x)$ is an irreducible polynomial over $\mathbb F_p$  and $\mathbb F_q$. It is not a costly task to find such a polynomial $f(x)$ of any degree. Once $f(x)$ is stated, the rest of the operations will take place in the ring $\mathbb Z_n$. The Jacobian group, Jac($N_n$) of $N$ over the ring $\mathbb Z_n$ is isomorphic to $$\text{Jac}(N_p)\oplus \text{Jac}(N_q).$$ Let assume that  $\deg(f(x))=r$ where $r$ is a positive integer. Then, the order of Jac($N_p)$, $ord_{N_p}$, is either $p^r+1$ or $p^r-1$. In a similar manner, Jac($N_q$),  $ord_{N_q}$, has order $q^r+1$ or $q^r-1$. 
Hence, the order of Jac($N_n)$ is  $$\mathcal K=ord_{N_p} \times ord_{N_q}$$

%Hence, the order of Jac($N_n)$ divides $$\mathcal K=(p^{2r}-1)(q^{2r}-1)$$

\noindent We can describe an element of Jac($N_n)$ as  $D=t(x)$  where $D = [f^2(x),t(x)f(x)]$ represents an element in $Jac(N_n)$ as we mention in the section above. Now we can say that $$t(x)^{s\mathcal K +1} =t(x) $$ in Jac($N_n)$ for any integer $s$.
In the method we proposed, as a first step a public encryption key $e$ is chosen such that $\gcd(e,\mathcal K)=1.$ Hence, we determine a private decryption key $d$ by calculating $$ed=1 \mod \mathcal K$$

\noindent This step of the key generation is given in Algorithm \ref{alg:our_key} below:

\begin{algorithm}[!h]
	\caption{Key Generation for Nodal Curve PKE Method}\label{alg:our_key}
	\begin{algorithmic}[1]
		\renewcommand{\algorithmicrequire}{\textbf{Input:}}
		\renewcommand{\algorithmicensure}{\textbf{Output:}}
		\REQUIRE $f(x), p$ and $q$, both prime and same bit-size
		\ENSURE Public Key $\gets (n,e)$, Private Key $\gets (n,d)$
		\STATE $n \gets p.q$
		\STATE $r \gets deg(f(x))$
		\STATE $ ord_{N_p}  \gets p^r-1$ or $p^r+1$
		\STATE $ ord_{N_q}  \gets  q^r-1$ or $q^r+1$
		\STATE $\mathcal K \gets  ord_{N_p} \times ord_{N_q}$
		\STATE Select $e; \qquad   \gcd(e, \mathcal K)=1;    1< e < \mathcal K $
		\STATE Compute $d; \qquad ed=1 \mod \mathcal K$ \\
		Public Key $\gets (n,e)$\\
		Private Key $\gets (n,d)$

	\end{algorithmic}
	
\end{algorithm}

\noindent In the second step, a message $m$ is embedded in a polynomial $t(x)$. This can be provided in the following way:
Suppose for a moment that $$t(x)=a_{r-1}x^{r-1}+\dots+a_1x+a_0$$
Then one can select $a_{r-1}$ randomly and assign blocks of the message $m$ to other coefficients $a_i$ where $i$ lies between 0 and $r-2$. In the following step, the cipher is obtained by encrypting the message $m$ as follows $$g(x)=t(x)^e$$ where we are taking $e^{th}$ power of $t(x)$ in the Jacobian of $N$ over the ring $\mathbb Z_n$.

Encryption and decryption parts of the method are given in Algorithm \ref{alg:our}. We should note here that, unlike the RSA algorithm, the cipher of $m$ would be distinct  each time that $m$ is encrypted. The first step of the encryption is the selection of a random coefficient of $t(x)$ and distinct $t(x)$s result in distinct ciphers and that is the reason we say the method is a probabilistic public-key algorithm. The next section includes security proofs and the comparison of experimental results between our method and the RSA algorithm.

\begin{algorithm}[H]
	\caption{Encryption and Decryption Steps for Nodal Curve PKE Method}\label{alg:our}
	\begin{algorithmic}[1]
		\renewcommand{\algorithmicrequire}{\textbf{Encryption:}}
		\renewcommand{\algorithmicensure}{\textbf{Decryption:}}
		\REQUIRE $E(m): $ \\
		$\quad\quad\quad\quad\qquad  a \gets$ a random integer \\
		$\quad\quad\quad\quad\qquad  m \gets {\{m_1,m_2,...,m_k\}} $  \\
		$\quad\quad\quad\quad\qquad  t(x)=ax^k+m_kx^{k-1}+m_{k-1}x^{k-2}+ ...+m_2x+m_1 \mod n;  k<r  $ \\
		$\quad\quad\quad\quad\qquad  g(x)=t(x)^e $ where $e^{th}$ power of $t(x)$ is computed in Jac($N_n$).\\
		$\quad\quad\quad\quad\qquad $The cipher $ c \gets g(x)$
		
		\ENSURE  $D(c): $ \\
		$\quad\quad\quad\quad\qquad t(x)=c^d=g(x)^d$ where $d^{th}$ power is taken in Jac($N_n$). \\
		$\quad\quad\quad\quad\qquad $The message  $m$ is obtained from $ t(x)$

	\end{algorithmic}
	
\end{algorithm}

\section{Analysis of the method}

  \subsection{Security Analysis}
  The workflow of the public-key algorithm with nodal curves is similar to RSA. A composite integer $n$ which is a multiplication of two prime integers $p$ and $q$ is one of the main ingredients. Unlike RSA, the group is not directly extracted from the multiplicative group $(\mathbb Z^*_n,\cdot)$, instead the employed group is obtained from a nodal curve $N:y^2=xf^2(x)$ over the ring $(\mathbb Z^*_n,+,\cdot)$. The generalized Jacobian group of the curve $N$ is the other component of ingredient for the algorithm. Proposing of such a Jacobian group is coming from the observation that the discrete logarithm problem (DLP) is hard on the groups where computing is handled in a similar manner. For example, DLP is assumed to be hard on the Jacobian groups of elliptic and hyperelliptic curves where computing in such groups is via Cantor's algorithm. In summary, the proposed public key is inspired by elliptic/hyperelliptic curve cryptography and RSA algorithm. Therefore, the security of aspect of the method is in some sense combination of discrete logarithm problem and integer factorization. In fact, an adversary has only the data flowing through public channels and this data includes the integer $n$, the curve $N$, and the encryption key $e$ in addition to the cipher $c$.  
  
  \begin{theorem} Let $(n, N, e)$ be the public key of user and $c$ be the cipher. If the factors of $n$ are known then one can compute the plaintext $m$.
  	
  	\end{theorem}
  
  \begin{proof}
  The curve   $N$ is defined by the equation $N:y^2=xf^2(x)$ where $f(x)$ is an irreducible polynomial over the ring $\mathbb Z_n[x]$. The plaintext $m$ is first embedded to an element of Jac($N$) by randomly selecting a polynomial $t(x)$ of degree less than $\deg f(x)$. The cipher  $c$  is just an element of Jac($N$) obtained by computing $t(x)^e$ in the Jacobian group. In order to go back from $c$ to $m$, one needs to know the order of Jac($N$). The order of Jac($N$) can be computed  if the factors of $n$ are available. In fact, assuming the factors are $p$ and $q$, the order becomes $\mathcal K=(p^{\deg f(x)}\pm 1) (q^{\deg f(x)}\pm 1)$. Once the order Jac($N$) is determined, it is easy task to find the decryption key $d$. In other words,
  $e^{-1} \mod \mathcal K$ returns $d$.
  
  \end{proof}
  	
 \begin{proposition}
 	Let $(n, N, e)$ and $c$ be as above. One can go back to the plaintext $m$ if $c^{\frac{1}{e}}$ can be computed in the Jacobian of $N$.
 \end{proposition}
  
  \begin{remark}
    RSA algorithm is based on the group $(\mathbb Z^*_n,\cdot)$ and the paper \cite{Naccache} indicates that finding an $e^{th}$ root of an element in this group can be done in a more efficient way than finding the factors of $n$. The representation of elements in the Jacobian of $N$ and computing in the group involves several polynomial arithmetic are the main factors that make DLP harder in such groups than in $(\mathbb Z_n^*,\cdot)$. For the similar reason with DLP, we believe computing $c^{\frac{1}{e}}$ in the Jac($N$) is much harder.
  \end{remark}
  
  \subsection{Performance Analysis}
  
  The above discussion leads to an intuition that employing nodal curves for a public-key algorithm gives more confidence from a security point of view. On the other hand, as computing in the Jacobian group requires several polynomial arithmetic, the advantage of using nodal curves seems to be disappeared. In other words, there is a trade-off between security and efficiency in the use of nodal curves. In addition, in case a solution for the RSA problem shows up, nodal curves themselves might stand as a candidate for a public-key algorithm in practice. In order to emphasize the practical performance of the algorithm, we present tests results  in the following tables (Table \ref{tab:enc} and Table \ref{tab:dec}). In the real-life usage of public-key algorithms, the clients are expected to perform encryption and the decryption is expected to be handled by the servers which are in general much more powerful machines than clients. Therefore, we keep the encryption key $e$ is small in both RSA and our proposed algorithm. 
    
  In all tests, we keep the primes $p$ and $q$ therefore $n$ same. In addition, the encryption key $e$ also stays the same for both algorithms. However, as the degree of $f(x)$ gets larger, $d$ grows exponentially in the algorithm with nodal curves and that is one of the important factors that the algorithm behaves much less efficiently than RSA algorithm in decryption phases. Computing in the Jacobian group requires polynomial arithmetic. In fact, each addition operation involves extended greatest common divisor algorithm of polynomials of degree less than $r=\deg f(x)$. The cost of this operation is bounded by $O(r^2)$ \cite{CohFrey} and as the operations take place over the ring $\mathbb Z_n$, a single addition in the Jacobian group costs $O(\log n r^2)$ bit operations. As a result, the cost of encryption is bounded by $O(r^2\log e\log n )$ whereas the encryption operation of RSA is bounded by $O(\log e \log n  )$. As for the decryption, while it stays as $O(r^2\log d' \log n)$ and $O(\log d\log n)$ for the proposed algorithm and RSA respectively, we should note here that the number $d'$ depends on $n$ and $r$ which makes the decryption operation more costly than RSA.

\begin{table}[H]
	
	\caption{Comparison of encryption phases for proposed algorithm with nodal curves and RSA algorithm. }
	\label{tab:enc}
	\begin{tabular}{ |p{2.5cm}|p{2.5cm}|p{3cm}|p{3cm}|  }

		%\multicolumn{3}{|c|}{} \\
		\hline
		{\bf Public Key Size}& 	{\bf Degree of the $f(x)$ }&{\bf Encryption with nodal curves (Second)} &{\bf Encryption with RSA (Second)}\\
		\hline
		\multirow{4}{2.5cm}{\centering  1024}  & \centering  2 & \centering  0.00269    &	\multirow{4}{3cm}{\centering  0.000034}\\ \cline{2-3} 
		
		& \centering 3& \centering  0.00379  &  \\ \cline{2-3}
		
		& \centering 4 &\centering   0.00418 &\\ \cline{2-3}
		
		&\centering  5 &\centering   0.00492 & \\ 
		\hline
		\multirow{4}{2.5cm}{ \centering  2048}  & \centering  2 & \centering  0.00402     &	\multirow{4}{3cm}{ \centering  0.000102}\\ \cline{2-3} 
		
		& \centering  3& \centering  0.00588  &  \\ \cline{2-3}
		
		& \centering  4 &\centering   0.00677 &\\ \cline{2-3}
		
		&\centering  5 &\centering   0.01238 & \\ 
		
		\hline
		
	\end{tabular}

\end{table}

\begin{table}[H]
	
	\caption{Comparison of decryption phases for the proposed algorithm and RSA algorithm. }
	
		\label{tab:dec}
	\begin{tabular}{ |p{2.5cm}|p{2.5cm}|p{3cm}|p{3cm}|  }

		\hline 
		{\bf Public Key Size}& 	{\bf Degree of the $f(x)$ } & {\bf Decryption with nodal curves (Second)} &{\bf Decryption with RSA (Second)}\\
		\hline
		\multirow{4}{2.5cm}{ \centering 1024}  & \centering 2 & \centering  0.4217    &	\multirow{4}{3cm}{ \centering 0.0032}\\ \cline{2-3} 
		
		& \centering 3&\centering 0.8905  &  \\ \cline{2-3}
		
		&\centering 4 &\centering 1.4045 &\\ \cline{2-3}
		
		&\centering 5 &\centering 2.0498 & \\ 
		\hline
		\multirow{4}{2.5cm}{\centering  2048}  & \centering 2 & \centering 1.3157     &	\multirow{4}{3cm}{\centering  0.0209}\\ \cline{2-3} 
		
		&\centering 3& \centering 2.8490  &  \\ \cline{2-3}
		
		&\centering  4 &\centering  4.5761 &\\ \cline{2-3}
		
		&\centering 5 &\centering  6.9647 & \\ 
		
		\hline
		
	\end{tabular}
	
\end{table}

   The experimental results were obtained while adapting the above computing method in the Jacobian groups which heavily requires polynomial arithmetic. One can avoid polynomial arithmetic for the nodal curves with a smaller arithmetic genus. For example,  if $\deg f(x)=2$ or $3$, one can perform group operation by conducting only integer arithmetic \cite[Chapter 14]{CohFrey}.

\begin{remark}

The work by D{\'e}ch{\`e}ne \cite{dechene2006} also suggests use of generalized Jacobian groups in cryptographic algorithms. Unlike our method, the suggested group operation is not based on Cantor’s method, and the security of suggested cryptographic algorithms is based on the hardness assumption of the discrete logarithm problem (DLP). On the other hand, the work \cite{galbraith} shows that the use of generalized Jacobian groups does not bring any advantage over Jacobian groups in respect to DLP. 
\end{remark}

 \bibliographystyle{amsplain}
 
 \bibliography{references}

\end{document}